\documentclass[12pt,a4paper]{amsart}

\usepackage[utf8]{inputenc}
\usepackage[american]{babel}
\usepackage{amsmath,amssymb,amsthm}
\usepackage{bm}
\usepackage{braket}
\usepackage{dsfont}
\usepackage{comment}
\usepackage[left=3cm,right=3cm]{geometry}
\usepackage[hidelinks]{hyperref}
\usepackage{cleveref}

\newcommand{\abs}[1]{\lvert#1\rvert}
\newcommand{\norm}[1]{\lVert#1\rVert}

\newcommand{\dom}{\operatorname{dom}}
\newcommand{\tr}{\operatorname{tr}}

\renewcommand{\Re}{\operatorname{Re}}
\renewcommand{\Im}{\operatorname{Im}}

\newtheorem{theorem}{Theorem}[section]
\newtheorem{proposition}[theorem]{Proposition}
\newtheorem{lemma}[theorem]{Lemma}
\newtheorem{corollary}[theorem]{Corollary}
\newtheorem*{conjecture}{Conjecture}

\theoremstyle{remark}
\newtheorem{remark}[theorem]{Remark}

\theoremstyle{definition}
\newtheorem{definition}[theorem]{Definition}

\title{Triangle Inequality for a Quantum Wasserstein Divergence}
\author{Melchior Wirth}
\address{Institute of Mathematics, Faculty of Mathematics and Computer Science, Leipzig University, Neues Augusteum, Augustusplatz 10, 04109 Leipzig, Germany}
\email{melchior.wirth@uni-leipzig.de}
\date{\today}

\begin{document}

\begin{abstract}
    We resolve a conjecture of De Palma and Trevisan by proving the triangle inequality for a quantum $2$-Wasserstein distance. The proof relies on complex analysis methods to establish a new integral representation of the cost in the optimal transport problem.
\end{abstract}

\maketitle

\section{Introduction}

Following the success of the theory of optimal transport in both pure and applied fields of mathematics in recent decades, significant research effort has been spent on developing a theory of noncommutative or quantum optimal transport and in particular on finding noncommutative analogs of the Wasserstein distances (see for example \cite{MRTV24} and \cite{Bea25} for some recent overviews).

Most approaches to define quantum Wasserstein distances of order $2$ can be classed in one of the following two groups. The first are static quantum Wasserstein distances such as the free Wasserstein distance of Biane and Voiculescu \cite{BV01} and the quantum Wasserstein distances of Ning, Georgiou and Tannenbaum \cite{NGT15}, Caglioti, Golse,
Mouhot, and Paul \cite{GMP16,CGP23}, De Palma and Trevisan \cite{DPT21,DPT24}, Friedland, Eckstein, Cole and \.Zyczkowski \cite{FECZ22} and Duvenhage \cite{Duv22,Duv23}. They define distances by a noncommutative analog of the Monge--Kantorovich formula
\begin{equation*}
    W_2(\mu,\nu)^2=\inf_\pi\int_{\mathbb R^d\times \mathbb R^d}\abs{x-y}^2\,d\pi(x,y),
\end{equation*}
where $\mu$, $\nu$ are Borel probability measures on $\mathbb R^d$ and the infimum runs over all Borel probability measures $\mu$ on $\mathbb R^d\times \mathbb R^d$ with marginals $\mu$ and $\nu$, called transport plans.

The second class are dynamical quantum Wasserstein distances such as the ones introduced by Carlen and Maas \cite{CM14,CM17}, Mielke and Mittnenzweig \cite{MM17} and Chen, Georgiou and Tannenbaum \cite{CGT18} for finite-dimensional quantum systems and extended to certain infinite-dimensional quantum systems by Hornshaw \cite{Hor18} and the author \cite{Wir21}. These distances are built on the model of the Benamou--Brenier formula
\begin{equation*}
    W_2(\mu,\nu)^2=\inf_{(\rho,v)}\int_0^1 \int_{\mathbb R^d}\abs{v_t}^2\,d\rho_t\,dt,
\end{equation*}
where $\mu$ and $\nu$ are again Borel probability measures on $\mathbb R^d$ and the infimum is taken over all pairs $\rho,v$ consisting of a curve $\rho$ of Borel probability measures and a curve $v$ of vector fields on $\mathbb R^d$ (subject to regularity conditions) that satisfy the continuity equation
\begin{equation*}
    \dot\rho_t+\nabla\cdot(v_t \rho_t)=0.
\end{equation*}

In quantum optimal transport, probability measures are replaced by density operators (or more generally states on a $C^\ast$-algebra or normal states on a von Neumann algebra). The transport plans in the Monge--Kantorovich formula are replaced tensor product states, quantum channels or bimodules, and for many distances, these three models are equivalent. Most static quantum Wasserstein distances can be classified according to the inner product in which the quadratic cost is expressed (GNS or KMS inner product) and the class of admissible transport plans.

In particular, the original quadratic Wasserstein distance defined by De Palma and Trevisan \cite{DPT21} can be expressed as
\begin{equation*}
    D(\rho,\sigma)=\inf_{\Phi}\sum_{k=1}^d \norm{x_k}_{\rho,\mathrm{GNS}}^2+\norm{x_k}_{\sigma,\mathrm{GNS}}^2-2\langle \Phi(x_k),x_k\rangle_{\sigma,\mathrm{KMS}},
\end{equation*}
where $x_1,\dots,x_d$ is a given set of self-adjoint operators and the infimum is taken over all duals of quantum channels such that $\Phi_\ast(\sigma)=\rho$.

One drawback of this definition is the distance $D$ has non-zero self-distances, which is a result of mixing the GNS and KMS inner product in the expression of the cost. To obtain an actual transport metric, De Palma and Trevisan proposed the following modified distance \cite{DPT24}:
\begin{equation*}
    W_{2,\mathrm{DPT}}(\rho,\sigma)=\sqrt{D(\rho,\sigma)^2-\frac 1 2D(\rho,\rho)^2-\frac 1 2 D(\sigma,\sigma)^2}.
\end{equation*}
We follow Bunth, Pitrik, Titkos and Virosztek in calling $W_{2,\mathrm{DPT}}$ the quantum Wasserstein divergence. This distance has zero self-distances by definition and is still symmetric, while non-degeneracy depends on the set of cost operators $x_1,\dots,x_d$. However, the question whether the distance $W_{2,\mathrm{DPT}}$ satisfies the triangle inequality was left open by De Palma and Trevisan.

\begin{conjecture}[De Palma, Trevisan {\cite[Remark 5.6]{DPT24}}]
    The map $W_{2,\mathrm{DPT}}$ satisfies the triangle inequality.
\end{conjecture}

Major progress on this conjecture was made by Bunth, Pitrik, Titkos and Virosztek \cite{BPTV24,BPTV25b,BPTV25}, who established the triangle inequality in the case when at least on the states involved is pure and for certain qubit states when the cost operators are given by the Pauli Matrices. They also presented numerical evidence that the triangle inequality holds for all states.

In this article, we show the triangle inequality for $W_{2,\mathrm{DPT}}$ in full generality without any restriction on the states (\Cref{cor:triangle_bounded}). Our proof relies on a new integral representation of the cost (\Cref{thm:integral_rep_cost}), which we establish using complex analysis methods.

For technical clarity, we first consider the case of bounded cost operators. However, unbounded cost operators are very relevant for applications in physics. For example, De Palma and Trevisan considered the position and momentum operators on $L^2(\mathbb R^d)$ as cost operators for the optimal transport between Gaussian states. The generalization of our results to the case of unbounded cost operators is straightforward and will be done in the last section.

\subsection*{Acknowledgments} The author is grateful to Daniel Virosztek for introducing him to the problem and sharing early drafts of \cite{BPTV25} with him. He also wants to thank Matthijs Vernooij for fruitful discussions on the closely related problem of representing completely Dirichlet forms in terms of derivations, which motivated this work.

\subsection*{Notation} If $H, K$ are Hilbert spaces, we write $\mathbb B(H;K)$ for the space of all bounded linear operators and $S^p(H;K)$ for the Schatten $p$-class. If $H=K$, we write $\mathbb B(H)$ and $S^p(H)$ instead of $\mathbb B(H;H)$ and $S^p(H;H)$. We call a linear map from $\mathbb B(H)$ to $\mathbb B(K)$ ucp if it is unital and completely positive.

\section{Triangle inequality for bounded cost operators}

In this section, we prove a new integral formula for the cost in the Wasserstein distance of De Palma and Trevisan (\Cref{thm:integral_rep_cost}), which gives rise to the subadditivity of the cost under compositions of quantum channels (\Cref{thm:subadditivity_cost}) and the triangle inequality (\Cref{cor:triangle_bounded}).

We start with a well-known characterization of normal ucp maps \cite[Theorem 9.2.3]{Dav76}.

\begin{proposition}[Kraus]
    Let $H$, $K$ be Hilbert spaces. A linear map $\Phi\colon \mathbb B(H)\to \mathbb B(K)$ is a normal ucp map if and only if there exists a family $(v_j)_{j\in J}$ in $\mathbb B(K;H)$ such that $\sum_{j\in J}v_j^\ast v_j=1$ in the strong operator topology and
    \begin{equation*}
        \Phi(x)=\sum_{j\in J}v_j^\ast x v_j
    \end{equation*}
    in the strong operator topology for all $x\in\mathbb B(H)$.

    Moreover, if $\rho\in S^1(H)$, $\sigma\in S^1(K)$, then $\Phi_\ast(\sigma)=\rho$ if and only if 
    \begin{equation*}
        \sum_{j\in J}v_j\sigma v_j^\ast=\rho
    \end{equation*}
    weakly in $S^1(H)$.
\end{proposition}

In the situation of Kraus' theorem, we call $\bm{v}=(v_j)_{j\in J}$ a family of Kraus operators for $\Phi$. We define
\begin{equation*}
    L_{\bm{v}}\colon S^2(K)\to \bigoplus_{j\in J}S^2(K;H),\,y\mapsto (v_j y)_{j\in J}.
\end{equation*}

\begin{lemma}
    If $H$, $K$ are Hilbert spaces and $\bm{v}=(v_j)_{j\in J}$ is a family in $\mathbb B(K;H)$ such that $\sum_{j\in J}v_j^\ast v_j=1$ in the strong operator topology, then $L_{\bm v}$ is an isometry.
\end{lemma}
\begin{proof}
    If $(e_\beta)_{\beta\in B}$ is an orthonormal basis of $K$ and $y\in S^2(K)$, then
    \begin{equation*}
        \norm{L_{\bm{v}}y}_2^2=\sum_{j\in J}\sum_{\beta\in B}\norm{v_j ye_\beta}^2=\sum_{\beta\in B}\sum_{j\in J}\norm{v_j y e_\beta}^2=\sum_{\beta\in B}\norm{ye_\beta}^2=\norm{y}_2^2.\qedhere
    \end{equation*}
\end{proof}

\begin{lemma}
    If $H$, $K$ are Hilbert spaces, $\rho\in S^1(H)$, $\sigma\in S^1(K)$ are positive and $\bm v=(v_j)_{j\in J}$ is a family in $\mathbb B(K;H)$ such that $\sum_{j\in J}v_j^\ast v_j=1$ in the strong operator topology and $\sum_{j\in J}v_j\sigma v_j^\ast=\rho$ weakly in $S^1(H)$, then there exists a unique bounded linear operator 
    \begin{equation*}
        R_{\bm v}\colon S^2(H)\to \bigoplus_{j\in J}S^2(K;H)
    \end{equation*}
    such that $R_{\bm v}(x\rho^{1/2})=(xv_j\sigma^{1/2})_{j\in J}$ for $x\in\mathbb B(H)$ and $R_{\bm v}=0$ on $(\mathbb B(H)\rho^{1/2})^\perp$.
\end{lemma}
\begin{proof}
    We have
    \begin{equation*}
        \sum_{j\in J}\norm{xv_j\sigma^{1/2}}_2^2=\sum_{j\in J}\tr(v_j \sigma v_j^\ast x^\ast x)=\tr(\rho x^\ast x)=\norm{x\rho^{1/2}}_2^2.\qedhere
    \end{equation*}
\end{proof}

The proof of the previous lemma shows that $R_{\bm v}$ is a partial isometry with initial space $\overline{\mathbb B(H)\rho^{1/2}}^{\norm{\cdot}_2}$.

\begin{lemma}
    Let $H$, $K$ be Hilbert spaces, $\rho\in S^1(H)$, $\sigma\in S^1(K)$ positive operators and $\Phi\colon \mathbb B(H)\to\mathbb B(K)$ a normal ucp map such that $\Phi_\ast(\sigma)=\rho$. If $\bm v=(v_j)$ is a family of Kraus operators for $\Phi$ and $x\in \mathbb B(H)$, then $L_{\bm v}^\ast R_{\bm v}(x\rho^{1/2})=\Phi(x)\sigma^{1/2}$ weakly in $S^2(H)$.
\end{lemma}
\begin{proof}
    If $y\in S^2(K)$, then
    \begin{equation*}
        \langle y,L_{\bm v}^\ast R_{\bm v}(x\rho^{1/2})\rangle_2=\sum_{j\in J}\tr(y^\ast v_j^\ast xv_j\sigma^{1/2})=\langle y,\Phi(x)\sigma^{1/2}\rangle.\qedhere
    \end{equation*}
\end{proof}

If $\rho\in S^1(H)$ is positive with spectral decomposition $\rho=\sum_{\alpha\in A}\lambda_\alpha \ket{e_\alpha}\bra{e_\alpha}$, where $(e_\alpha)_{\alpha\in A}$ is an orthonormal basis of $H$, and $z\in\mathbb C$, we define
\begin{equation*}
    \rho^z=\sum_{\alpha:\lambda_\alpha>0}\lambda_\alpha^z \ket{e_\alpha}\bra{e_\alpha}.
\end{equation*}
If $\Re z<0$, the operator $\rho^z$ has to be understood as unbounded operator with domain $\dom(\rho^z)=\{\xi\in H\mid \sum_{\alpha:\lambda_\alpha>0}\lambda^{2\Re z}\abs{\langle e_\alpha,\xi\rangle}^2<\infty\}$.

Note that if $p>0$ and $t\in\mathbb R$, then $\rho^{1/p+it}\in S^p(H)$ with $\norm{\rho^{1/p+it}}_p^p=\norm{\rho}_1$.

We define the \emph{KMS inner product} on $\mathbb B(H)$ by
\begin{equation*}
    \langle x,y\rangle_\rho=\tr(x^\ast\rho^{1/2}y\rho^{1/2}).
\end{equation*}
Note that this form is degenerate if $\rho$ has non-trivial kernel, and thus it is not a proper inner product in general.

\begin{lemma}\label{lem:cyclic_space}
    If $\rho\in S^1(H)$ is a positive operator, then $\overline{\mathbb B(H)\rho^{1/2}}^{\norm\cdot_2}=\{y\in S^2(H)\mid \ker \rho\subset \ker y\}$. In particular, if $x\in \mathbb B(H)$, then $\rho^{1/4+it}x\rho^{1/4-it}\in \overline{\mathbb B(H)\rho^{1/2}}^{\norm{\cdot}_2}$ for all $t\in\mathbb R$.
\end{lemma}
\begin{proof}
    Clearly, $\ker\rho=\ker\rho^{1/2}\subset\ker(x\rho^{1/2})$ and the inclusion of kernels is preserved under limits in the Hilbert--Schmidt norm. Conversely, if $y\in S^2(H)$ and $\ker\rho\subset \ker y$, let $x_n e_\alpha=0$ if $\lambda_\alpha<1/n$ and $x_n e_\alpha=\lambda_\alpha^{-1/2}y e_\alpha$ if $\lambda_\alpha\geq 1/n$. The operator $x_n$ is bounded with $\norm{x_n}\leq\sqrt n \norm{y}$. Moreover,
    \begin{equation*}
        \norm{x_n\rho^{1/2}-y}_2^2=\sum_{\alpha:0<\lambda_\alpha<\frac 1 n}\norm{y e_\alpha}^2\to 0
    \end{equation*}
    as $n\to\infty$ by the dominated convergence theorem. The second claim is now clear.
\end{proof}

\begin{proposition}
    Let $H$, $K$ be Hilbert spaces, $\rho\in S^1(H)$, $\sigma\in S^1(K)$ positive operators and let $x\in\mathbb B(H)$, $y\in \mathbb B(K)$ be self-adjoint. If $\Phi\colon \mathbb B(H)\to \mathbb B(K)$ is a normal ucp map such that $\Phi_\ast(\sigma)=\rho$ and $\bm v=(v_j)$ is a family of Kraus operators for $\Phi$, then there exists a bounded continuous function $f\colon\{z\in\mathbb C:\abs{\Re z}\leq 1/4\}\to\mathbb C$ that is analytic on the interior and satisfies
    \begin{align*}
        f(0)&=\langle\Phi(x),y\rangle_\sigma,\\
        f(1/4+it)&=\langle L_{\bm v}(\sigma^{1/4+it}y\sigma^{1/4-it}),R_{\bm v}(\rho^{1/4+it}x\rho^{1/4-it})\rangle_2,\\
        f(-1/4+it)&=\langle R_{\bm v}(\rho^{1/4+it}x\rho^{1/4-it}),L_{\bm v}(\sigma^{1/4+it}y\sigma^{1/4-it})\rangle_2
    \end{align*}
    for all $t\in\mathbb R$.
\end{proposition}
\begin{proof}
    If $0\leq \Im z\leq 1/4$ let
    \begin{equation*}
        F(z)=\langle L_{\bm v}(\sigma^{1/2-i\bar z} y \sigma^{i\bar z}),R_{\bm v}(\rho^{-iz} x \rho^{1/2+iz})\rangle_2.
    \end{equation*}
    Since all the exponents have nonnegative real part, all the operators involved are well-defined. Moreover, if $\Im z=s$, then
    \begin{align*}
        \norm{L_{\bm v}(\sigma^{1/2-i\bar z} y \sigma^{i\bar z})}_2&=\norm{\sigma^{1/2-i\bar z} y \sigma^{i\bar z}}_2\\
        &\leq \norm{\sigma^{1/2-i\bar z}}_{(1/2-s)^{-1}}\norm{\sigma^{i\bar z}}_{s^{-1}}\norm{y}\\&=\norm{\sigma}_1^{1/2-s}\norm{\sigma}_1^s\norm{y}\\
        &=\norm{\sigma}_1^{1/2}\norm{y}
    \end{align*}
    and
    \begin{align*}
        \norm{R_{\bm v}(\rho^{-iz} x \rho^{1/2+iz})}_2&\leq \norm{\rho^{-iz}x\rho^{1/2+iz}}_2\\
        &\leq \norm{\rho^{-iz}}_{s^{-1}}\norm{\rho^{1/2+iz}}_{(1/2-s)^{-1}}\norm{x}\\
        &=\norm{\rho}_1^{1/2}\norm{x},
    \end{align*}
    with the obvious modifications if $s=0$.

    Thus
    \begin{equation*}
        \abs{F(z)}\leq \norm{\rho}_1^{1/2}\norm{\sigma}_1^{1/2}\norm{x}\norm{y}
    \end{equation*}
    for all $z\in\mathbb C$ with $0\leq \Im z\leq 1/4$. By the spectral theorem, $F$ is continuous on the strip $\{z\in\mathbb C\mid 0\leq \Im z\leq 1/4\}$ and analytic on the interior.

    Moreover, if $t\in\mathbb R$, then
    \begin{align*}
        F(t)&=\langle L_{\bm v}(\sigma^{1/2-it}y\sigma^{it}),R_{\bm v}(\rho^{-it}x\rho^{1/2+it})\rangle_2\\
        &=\langle \sigma^{1/2-it}y\sigma^{it},\Phi(\rho^{-it}x\rho^{it})\sigma^{1/2}\rangle_2\\
        &=\langle \sigma^{-it}y\sigma^{it},\Phi(\rho^{-it}x\rho^{it})\rangle_\sigma.
    \end{align*}
    In particular, $F(t)\in\mathbb R$.

    By the Schwarz reflection principle, $F$ extends to a bounded continuous function on the strip $\{z\in \mathbb C\mid -1/4\leq \Im z\leq 1/4\}$ that is analytic on the interior and satisfies $F(\bar z)=\overline{F(z)}$ for all $z\in \mathbb C$ with $\abs{\Im z}\leq 1/4$. Now $f(z)=F(iz)$ for $z\in\mathbb C$ with $\abs{\Re z}\leq 1/4$ defines a function with the desired properties.
\end{proof}

\begin{theorem}[Integral representation of the cost]\label{thm:integral_rep_cost}
    Let $H$, $K$ be Hilbert spaces, $\rho\in S^1(H)$, $\sigma\in S^1(K)$ positive operators and let $x\in\mathbb B(H)$, $y\in \mathbb B(K)$ be self-adjoint. If $\Phi\colon \mathbb B(H)\to \mathbb B(K)$ is a normal ucp map such that $\Phi_\ast(\sigma)=\rho$ and $\bm v=(v_j)$ is a family of Kraus operators for $\Phi$, then
    \begin{align*}
        \norm{x}_\rho^2+\norm{y}_\sigma^2-2\langle \Phi(x),y\rangle_\sigma=\int_{\mathbb R}\norm{L_{\bm v}(\sigma^{1/4+it}y\sigma^{1/4-it})-R_{\bm v}(\rho^{1/4+it}x\rho^{1/4-it})}_2^2\,d\mu(t),
    \end{align*}
    where $\mu$ is the Borel probability measure on $\mathbb R$ with density $\frac{d\mu}{dt}= 2(\cosh 2\pi t)^{-1}$.
\end{theorem}
\begin{proof}
    As noted above, $\rho^{1/4+it}\in S^4(H)$ and $\sigma^{1/4+it}\in S^4(K)$, so the expression on the right side is well-defined.
    
    Since $L_{\bm v}$ is an isometry, we have
    \begin{align*}
        \norm{L_{\bm v}(\sigma^{1/4+it}y\sigma^{1/4-it})}_2^2=\norm{\sigma^{1/4+it}y\sigma^{1/4-it}}_2^2=\norm{\sigma^{1/4}y\sigma^{1/4}}_2^2=\norm{y}_\sigma^2,
    \end{align*}
    where we used that $\sigma^{it}$ acts as a unitary on $(\ker \sigma)^\perp$.
    
    Moreover, since $\rho^{1/4+it}x\rho^{1/4-it}$ belongs to the initial space of $R_{\bm v}$ by the previous lemma, $\norm{R_{\bm v}(\rho^{1/4+it}x\rho^{1/4-it})}_2^2=\norm{x}_\rho^2$ by the same argument.

    To evaluate the cross term, note that the previous proposition guarantees the existence of a bounded continuous function $f\colon \{z\in\mathbb C:\abs{\Re z}\leq 1/4\}\to \mathbb C$ that is analytic on the interior and satisfies $f(0)=\langle \Phi(x),y\rangle_\sigma$ and 
    \begin{equation*}
        f(1/4+it)+f(-1/4+it)=2\Re\langle L_{\bm v}(\sigma^{1/4+it}y\sigma^{1/4-it}),R_{\bm v}(\rho^{1/4+it}x\rho^{1/4-it})\rangle.
    \end{equation*}
    For $z\in\mathbb C$ with $\abs{\Re z}\leq\frac1  4$ let $g(z)=\frac{f(z)}{\sin 2\pi z}$. The function $g$ is continuous, bounded and meromorphic on the interior of its domain with a single pole with residue $f(0)/2\pi$ at $0$. Moreover, $\lim_{R\to\infty}\abs{g(s\pm iR)}=0$ uniformly in $s\in [-1/4,1/4]$. By the residue theorem,
    \begin{equation*}
        f(0)=\frac 1 2\int_{\mathbb R}(f(1/4+it)+f(-1/4+it))\,d\mu(t).\qedhere
    \end{equation*}
\end{proof}

\begin{remark}
    Informally, the right side of the integral representation is
    \begin{equation*}
        \sum_{j\in J}\int_{\mathbb R}\norm{v_j \sigma^{1/4+it}y\sigma^{1/4-it}-\rho^{1/4+it}x\rho^{-1/4-it}v_j\sigma^{1/2}}_2^2\frac{2\,dt}{\cosh 2\pi t},
    \end{equation*}
    but $\rho^{-1/4-it}$ is typically not bounded, so some care has to be taken interpreting this expression (which is achieved by the definition of the map $R_{\bm v}$).
\end{remark}

\begin{remark}
    If $\rho=\sigma$ is non-singular and $\Phi$ is symmetric with respect to the KMS inner product, then the (closure of) the quadratic form
    \begin{equation*}
        \rho^{1/4}\mathbb B(H)\rho^{1/4}\to [0,\infty),\,\rho^{1/4}x\rho^{1/4}\mapsto \norm{x}_{\rho}^2-\langle x,\Phi(x)\rangle_{\rho}
    \end{equation*}
    is a bounded conservative completely Dirichlet form. Hence this integral representation is closely related to the problem of expressing completely Dirichlet forms in terms of derivations. This connection will be further explored in follow-up work with Matthijs Vernooij.
\end{remark}

\begin{definition}[Cost]
Let $H$, $K$ be Hilbert spaces, $\rho\in S^1(H)$, $\sigma\in S^1(K)$ positive operators and let $\bm x\in \mathbb B(H)^d$, $\bm y\in \mathbb B(K)^d$ be $d$-tuples of self-adjoint operators. If $\Phi\colon \mathbb B(H)\to \mathbb B(K)$ is a normal ucp map such that $\Phi_\ast(\sigma)=\rho$, then its \emph{cost} is defined as
\begin{equation*}
    C_{\rho,\sigma,\bm x,\bm y}(\Phi)=\sum_{k=1}^d(\norm{x_k}_\rho^2+\norm{y_k}_\sigma^2-2\langle\Phi(x_k),y_k\rangle_\sigma).
\end{equation*}
\end{definition}

\begin{theorem}[Subadditivity of the cost]\label{thm:subadditivity_cost}
    For $k\in \{1,2,3\}$ let $H_k$ be a Hilbert space, $\rho_k\in S^1(H_k)$ a positive operator and $\bm{x}_k\in \mathbb B(H_k)^d$ a $d$-tuple of self-adjoint operators. If $\Phi_{12}\colon \mathbb B(H_1)\to \mathbb B(H_2)$ and $\Phi_{23}\colon\mathbb B(H_2)\to\mathbb B(H_3)$ are normal ucp maps such that $\Phi_{12\ast}(\rho_2)=\rho_1$ and $\Phi_{23\ast}(\rho_3)=\rho_2$, then
    \begin{equation*}
        C_{\rho_1,\rho_3,\bm x_1,\bm x_3}(\Phi_{23}\circ \Phi_{12})^{1/2}\leq C_{\rho_1,\rho_2,\bm x_1,\bm x_2}(\Phi_{12})^{1/2}+C_{\rho_2,\rho_3,\bm x_2,\bm x_3}(\Phi_{23})^{1/2}.
    \end{equation*}
\end{theorem}
\begin{proof}
    For ease of notation, we drop the subscripts of the cost in the following. We can assume that $d=1$, since the general case then follows from the triangle inequality for the Euclidean distance on $\mathbb R^d$. Let $\bm v=(v_i)_{i\in I}$ be a family of Kraus operators for $\Phi_{12}$ and $\bm w=(w_j)_{j\in J}$ a family of Kraus operators for $\Phi_{23}$. Clearly, $\bm v\otimes \bm w=(v_i w_j)_{(i,j)\in I\times J}$ is a family of Kraus operators for $\Phi_{23}\circ \Phi_{12}$.

    Note that $L_{\bm v\otimes\bm w}=(\bigoplus_j L_{\bm v})\circ L_{\bm w}$ and $R_{\bm v\otimes \bm w}=(\bigoplus_i R_{\bm w})\circ R_{\bm v}$. By \Cref{thm:integral_rep_cost},
    \begin{align*}
        C(\Phi_{23}\circ\Phi_{12})=\int_{\mathbb R}\norm{L_{\bm v\otimes\bm w}(\rho_3^{1/4+it}x_3\rho_3^{1/4+it})-R_{\bm v\otimes\bm w}(\rho_1^{1/4+it}x_1\rho_1^{1/4+it})}_2^2\,d\mu(t).
    \end{align*}
    Since $L_{\bm v}$ is an isometry and $R_{\bm w}$ is a partial isometry, we have
    \begin{align*}
        &\norm{L_{\bm v\otimes \bm w}(\rho_3^{1/4+it}x_3\rho_3^{1/4+it})-(\oplus_j L_{\bm v})R_{\bm w}(\rho_2^{1/4+it}x_2\rho_2^{1/4+it})}_2\\
        &\qquad=\norm{L_{\bm w}(\rho_3^{1/4+it}x_3\rho_3^{1/4+it})-R_{\bm w}(\rho_2^{1/4+it}x_2\rho_2^{1/4+it})}_2
    \end{align*}
    and
    \begin{align*}
        &\norm{(\oplus_i R_{\bm w})L_{\bm v}(\rho_2^{1/4+it}x_2\rho_2^{1/4+it})-R_{\bm v\otimes\bm w}(\rho_1^{1/4+it}x_1\rho_1^{1/4+it})}_2\\
        &\qquad\leq \norm{L_{\bm v}(\rho_2^{1/4+it}x_2\rho_2^{1/4+it})-R_{\bm v}(\rho_1^{1/4+it}x_1\rho_1^{1/4+it})}_2.
    \end{align*}
    Note further that $(\bigoplus_j L_{\bm v})\circ R_{\bm w}=(\bigoplus_i R_{\bm w})\circ L_{\bm v}$ directly from the definition of these maps. Now the claimed subadditivity for the cost follows from the triangle inequality for the Hilbert--Schmidt norm and the triangle inequality in $L^2(\mathbb R,\mu)$.
\end{proof}

\begin{definition}[Wasserstein divergence]
Let $H$ be a Hilbert space and $\bm x\in \mathbb B(H)^d$ a $d$-tuple of self-adjoint operators. If $\rho,\sigma\in S^1(H)$ are positive trace-class operators with trace $1$, their \emph{Wasserstein divergence} is defined as
\begin{equation*}
    W_{2,\mathrm{DPT}}(\rho,\sigma)=\inf\{C_{\rho,\sigma,\bm x,\bm y}(\Phi)^{1/2}\mid \Phi\colon \mathbb B(H)\to\mathbb B(H)\text{ normal ucp, }\Phi_\ast(\sigma)=\rho\}.
\end{equation*}
\end{definition}

\begin{remark}
   For the fact that this definition of the distance $W_{2,\mathrm{DPT}}$ with the original one given by De Palma and Trevisan, see \cite[Eq. (95)]{BPTV25}.
\end{remark}

\begin{corollary}\label{cor:triangle_bounded}
    The Wasserstein divergence $W_{2,\mathrm{DPT}}$ satisfies the triangle inequality.
\end{corollary}

\section{Extension to unbounded costs}

In this section we extend the results from the previous section to the case when the cost operators are unbounded.

\begin{lemma}\label{lem:ucp_KMS_bounded}
    Let $H$, $K$ be Hilbert spaces, $\rho\in S^1(H)$, $\sigma\in S^1(K)$ positive operators and $\Phi\colon \mathbb B(H)\to\mathbb B(K)$ a normal ucp map such that $\Phi_\ast(\sigma)=\rho$. If $x\in \mathbb B(H)$, then $\norm{\sigma^{1/4}\Phi(x)\sigma^{1/4}}_2\leq \norm{\rho^{1/4}x\rho^{1/4}}_2$.
\end{lemma}
\begin{proof}
    By the Kadison--Schwarz inequality,
    \begin{equation*}
        \norm{\Phi(x)\sigma^{1/2}}_2^2=\tr(\Phi(x)^\ast \Phi(x)\sigma)\leq \tr(x^\ast x\Phi_\ast(\sigma))=\norm{x\rho^{1/2}}_2^2.
    \end{equation*}
    Likewise, $\norm{\sigma^{1/2}\Phi(x)}_2\leq \norm{\rho^{1/2}x}_2$. Now the result follows from the interpolation theory of Hilbert spaces (see \cite{Don67}).
\end{proof}

\begin{definition}[Finite second moments]
    Let $H$ be a Hilbert space and $\bm x=(x_k)_{k=1}^d$ a $d$-tuple of self-adjoint operators on $H$. We say that a positive operator $\rho\in S^1(H)$ has \emph{finite second moments} if $\rho^{1/4}H\subset \dom(x_k)$, $\rho^{1/4}x_k \rho^{1/4}$ is closable and $\overline{\rho^{1/4}x_k\rho^{1/4}}\in S^2(H)$ for $1\leq k\leq d$. In this case, we write $i_\rho(x_k)$ for $\overline{\rho^{1/4}x_k\rho^{1/4}}$.
\end{definition}

Note that if $\rho$ has finite second moments, then 
\begin{equation*}
    (\rho^{1/4}x_k\rho^{1/4})^\ast\supset \rho^{1/4}(\rho^{1/4}x_k)^\ast=\rho^{1/4}x_k\rho^{1/4}
\end{equation*}
by Schmüdgen, Proposition 1.7. Hence $i_\rho(x_k)$ is self-adjoint.

\begin{lemma}\label{lem:KMS_embedding}
    If $\rho\in S^1(H)$ is positive, then
    \begin{equation*}
        \overline{\{\rho^{1/4}x\rho^{1/4}\mid x=x^\ast\in\mathbb B(H)\}}^{\norm\cdot_2}=\{y\in S^2(H)\mid y=y^\ast,\ker\rho\subset\ker y\}.
    \end{equation*}
\end{lemma}
\begin{proof}
    The proof is similar to that of \Cref{lem:cyclic_space}. Again, if $x\in \mathbb B(H)$, then $\ker \rho=\ker \rho^{1/4}\subset \ker(\rho^{1/4}x\rho^{1/4})$, and this relation is preserved in the limit, which settles one inclusion.

    For the opposite inclusion let $\rho=\sum_{\alpha\in A}\lambda_\alpha\ket{e_\alpha}\bra{e_\alpha}$ be the spectral decomposition of $\rho$ and let $p_n$ be the orthogonal projection onto the closed linear span of $\{e_\alpha\mid \alpha\in A,\,\lambda_\alpha\in \{0\}\cup[1/n,\infty)\}$. If $y\in S^2(H)$ is self-adjoint, then the operator $x_n=\rho^{-1/4}p_n y \rho^{-1/4}p_n$ is bounded and self-adjoint. Moreover, if $\ker\rho\subset\ker y$, then $\overline{\operatorname{ran} y}\subset (\ker \rho)^\perp$ and hence
    \begin{align*}
        \norm{(\rho^{1/4}x_n \rho^{1/4}-y)e_\alpha}=\begin{cases}
            0&\text{if }\lambda_\alpha\in \{0\}\cup[\frac 1 n,\infty),\\
            \norm{ye_\alpha}&\text{if }\lambda_\alpha\in (0,\frac 1 n).
        \end{cases}
    \end{align*}
    Thus $\rho^{1/4}x_n\rho^{1/4}\to y$ in Hilbert--Schmidt norm by the dominated convergence theorem.
\end{proof}

Let $\rho\in S^1(H)$, $\sigma\in S^1(K)$ be positive operators and $\Phi\colon \mathbb B(H)\to\mathbb B(K)$ is a normal ucp map. By \Cref{lem:ucp_KMS_bounded}, the map
\begin{equation*}
    \rho^{1/4}\mathbb B(H)\rho^{1/4}\to\sigma^{1/4}\mathbb B(K)\sigma^{1/4},\,\rho^{1/4}x\rho^{1/4}\mapsto \sigma^{1/4}\Phi(x)\sigma^{1/4}
\end{equation*}
extends to a contractive linear map 
\begin{equation*}
    \Phi^{(2)}\colon\overline{\rho^{1/4}\mathbb B(H)\rho^{1/4}}^{\norm\cdot_2}\to\overline{\sigma^{1/4}\mathbb B(K)\sigma^{1/4}}^{\norm\cdot_2}.
\end{equation*}

\begin{definition}[Cost]
    Let $\bm x=(x_k)_{k=1}^d$ and $\bm y=(y_k)_{k=1}^d$ be $d$-tuples of self-adjoint operators on $H$ and $K$, respectively, and let $\rho\in S^1(H)$ and $\sigma\in S^1(K)$ be positive operators with finite second moments. If $\Phi\colon \mathbb B(H)\to\mathbb B(K)$ is a normal ucp map, we define its \emph{cost} as
    \begin{equation*}
    C_{\rho,\sigma,\bm x,\bm y}(\Phi)=\sum_{k=1}^d\norm{i_\rho(x_k)}_2^2+\norm{i_\sigma(y_k)}_2^2-2\langle \Phi^{(2)}(i_\rho(x_k)),i_\sigma(y_k)\rangle_2.
    \end{equation*}
\end{definition}
\Cref{lem:KMS_embedding} guarantees that $i_\rho(x_k)$ belongs to the domain of definition of $\Phi^{(2)}$. Clearly, this definition is consistent with the previous definition for bounded cost operators.

The integral representation of the cost now readily generalizes to the case of unbounded cost operators.

\begin{proposition}
Let $H$, $K$ be Hilbert spaces, $\bm x=(x_k)_{k=1}^d$ and $\bm y=(y_k)_{k=1}^d$ $d$-tuples of self-adjoint operators on $H$ and $K$, respectively, and let $\rho\in S^1(H)$, $\sigma\in S^1(K)$ be positive operators with finite second moments. If $\Phi\colon \mathbb B(H)\to \mathbb B(K)$ is a normal ucp map such that $\Phi_\ast(\sigma)=\rho$ and $\bm v=(v_j)$ is a family of Kraus operators for $\Phi$, then
    \begin{equation*}
        C_{\rho,\sigma,\bm x,\bm y}(\Phi)=\sum_{k=1}^d\int_{\mathbb R}\norm{L_{\bm v}(\sigma^{it}i_\sigma(y_k)\sigma^{-it})-R_{\bm v}(\rho^{it}i_\rho(x_k)\rho^{-it})}_2^2\,d\mu(t),
    \end{equation*}
    where $\mu$ is the Borel probability measure on $\mathbb R$ with density $\frac{d\mu}{dt}= 2(\cosh 2\pi t)^{-1}$.
\end{proposition}
\begin{proof}
    By \Cref{lem:KMS_embedding}, there exist sequences $(x_k^{(n)})_n$ and $(y_k^{(n)})_n$ of bounded self-adjoint operators on $H$ and $K$, respectively, such that $\rho^{1/4}x_k^{(n)}\rho^{1/4}\to i_\rho(x_k)$ and $\sigma^{1/4}y_k^{(n)}\sigma^{1/4}\to i_\sigma(y_k)$ in Hilbert--Schmidt norm as $n\to\infty$. Then the result follows from \Cref{thm:integral_rep_cost} in the limit $n\to\infty$.
\end{proof}

The same proof as for \Cref{thm:subadditivity_cost} shows the subadditivity of the cost.

\begin{corollary}
    For $k\in \{1,2,3\}$ let $H_k$ be a Hilbert space, $\bm{x}_k$ a $d$-tuple of self-adjoint operators on $H_K$ and $\rho_k\in S^1(H_k)$ a positive operator with finite second moments. If $\Phi_{12}\colon \mathbb B(H_1)\to \mathbb B(H_2)$ and $\Phi_{23}\colon\mathbb B(H_2)\to\mathbb B(H_3)$ are normal ucp maps such that $\Phi_{12\ast}(\rho_2)=\rho_1$ and $\Phi_{23\ast}(\rho_3)=\rho_2$, then
    \begin{equation*}
        C_{\rho_1,\rho_3,\bm x_1,\bm x_3}(\Phi_{23}\circ \Phi_{12})^{1/2}\leq C_{\rho_1,\rho_2,\bm x_1,\bm x_2}(\Phi_{12})^{1/2}+C_{\rho_2,\rho_3,\bm x_2,\bm x_3}(\Phi_{23})^{1/2}.
    \end{equation*}
\end{corollary}

\begin{definition}[Wasserstein divergence]
    Let $H$ be a Hilbert space and $\bm x$ a $d$-tuple of self-adjoint operators on $H$. If $\rho,\sigma\in S^1(H)$ are positive trace-class operators with trace $1$, their \emph{Wasserstein divergence} is defined as
\begin{equation*}
    W_{2,\mathrm{DPT}}(\rho,\sigma)=\inf\{C_{\rho,\sigma,\bm x,\bm y}(\Phi)^{1/2}\mid \Phi\colon \mathbb B(H)\to\mathbb B(H)\text{ normal ucp, }\Phi_\ast(\sigma)=\rho\}.
\end{equation*}
\end{definition}

\begin{corollary}
    The Wasserstein divergence satisfies the triangle inequality.
\end{corollary}

\bibliographystyle{alpha}
\bibliography{ref}

\end{document}